\documentclass[a4paper]{article}

\usepackage[margin=1.38in,top=1.1in]{geometry}

\usepackage[T1]{fontenc}
\usepackage[latin9]{inputenc}
\usepackage{amsmath, dsfont}
\usepackage{amsmath}
\usepackage{amssymb}
\usepackage{amsthm}
\usepackage{wasysym}
\usepackage{graphicx}
\usepackage{thmtools}
\usepackage{thm-restate}
\usepackage{hyperref}

\usepackage{enumerate}
\usepackage{verbatim}

\newtheorem{lemma}{Lemma}

\newtheorem{corollary}[lemma]{Corollary}

\newtheorem{theorem}[lemma]{Theorem}

\newtheorem{open}[lemma]{Open Question}

\usepackage{thmtools}

\graphicspath{{figures/}}

\begin{document}
\thispagestyle{empty}

\begin{center}
\huge 
Approximation and Hardness for \\ Token Swapping
\end{center}

\vspace{0.5cm}

\begin{minipage}[c]{0.3\textwidth}%
  \begin{center}
  \textsc{Tillmann Miltzow}\\ [0.1cm]
  {\small Freie Universit\"at Berlin \\ }
  {\small Berlin, Germany \\[0.4cm] }
  \par\end{center}%
\end{minipage}
\begin{minipage}[c]{0.3\textwidth}%
  \begin{center}
  \textsc{Lothar Narins}\\ [0.1cm]
  {\small Freie Universit\"at Berlin \\ }
  {\small Berlin, Germany \\[0.4cm] }
  \par\end{center}%
\end{minipage}
\begin{minipage}[c]{0.3\textwidth}%
  \begin{center}
  \textsc{Yoshio Okamoto}\\ [0.1cm]
  {\small University of Electro-Communications \\ }
  {\small Tokyo, JAPAN \\[0.4cm] }
  \par\end{center}%
\end{minipage}

\medskip

\begin{minipage}[c]{0.3\textwidth}%
  \begin{center}
  \textsc{G\"unter Rote}\\ [0.1cm]
  {\small Freie Universit\"at Berlin \\ }
  {\small Berlin, Germany \\[0.4cm] }
  \par\end{center}%
\end{minipage}
\begin{minipage}[c]{0.3\textwidth}%
  \begin{center}
  \textsc{Antonis Thomas}\\ [0.1cm]
  {\small  Department of Computer Science, ETH Z\"{u}rich \\ }
  {\small Switzerland \\[0.4cm] }
  \par\end{center}%
\end{minipage}
\begin{minipage}[c]{0.3\textwidth}%
  \begin{center}
  \textsc{Takeaki Uno}\\ [0.1cm]
  {\small  National Institute of Informatics \\ }
  {\small Tokyo, JAPAN \\[0.4cm] }
  \par\end{center}%
\end{minipage}

\vspace{0.5cm}

\begin{figure}[htbp]
  \centering\includegraphics[width = 0.8\textwidth]{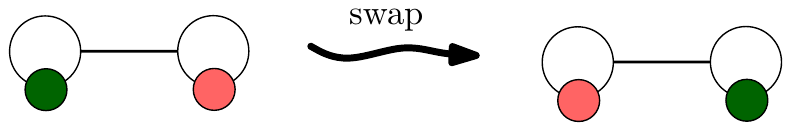}
  
  Two tokens on adjacent vertices swap
\end{figure}

\vspace{0.2cm}

\begin{abstract}

Given a graph $G=(V,E)$ with $V=\{1,\ldots,n\}$, we
place on every vertex a token $T_1,\ldots,T_n$. A swap is an exchange of
tokens on adjacent vertices. We consider the algorithmic question of
finding a shortest sequence of swaps such that token $T_i$ is on 
vertex $i$.
We are able to achieve essentially matching upper and lower bounds, for exact algorithms and approximation algorithms. 
For exact algorithms, we rule out any $2^{o(n)}$ algorithm under the ETH. This is matched with a simple $2^{O(n\log n)}$ algorithm based on a breadth-first search in an auxiliary graph.
We show one general $4$-approximation and show APX-hardness.
Thus, there is a small constant $\delta>1$ such that every polynomial time approximation algorithm has approximation factor at least $\delta$.

Our results also hold for a generalized version, where tokens 
and vertices are colored. 
In this generalized version each token 
must go to a vertex with the same color.
\end{abstract}

{\let\thefootnote\relax
\footnotetext
{A condensed version of this report will be presented at
the
 24th Annual European Symposium on Algorithms in Aarhus, Denmark,
in August 2016~\cite{mnortu-ahts-16}.}
}




\section{Introduction}

In the theory of computation, we regularly encounter the following type of 
problem:
Given two configurations, we wish to transform one to the other.
In these problems we also need to fix a family of operations that we 
are allowed to perform.
Then, we need to solve two problems: (1) Determine if one can be
transformed to the other; (2) If so, find a shortest sequence of 
such operations.
Motivations come from the better understanding of solution spaces,
which is beneficial for design of local-search algorithms, enumeration,
and probabilistic analysis.
See \cite{GKMP09} for examples.
The study of \emph{combinatorial
reconfigurations} is a young growing field  \cite{CoRe2015}.

Among problem variants in combinatorial reconfiguration,
we study the \emph{token swapping problem} on a graph.
The problem is defined as follows.
We are given an undirected connected graph with $n$ vertices $v_1,\dots,v_n$.
Each vertex $v_i$ holds exactly one token $T_{\pi(i)}$, where $\pi$ is a
permutation of $\{1,\dots,n\}$.
In one step, we are allowed to swap tokens on a pair of adjacent vertices, 
that is, if $v_i$ and $v_j$ are adjacent, $v_i$ holds the token $T$, and
$v_j$ holds the token $T'$, then the swap between $v_i$ and $v_j$ results
in the configuration where $v_i$ holds $T'$, $v_j$ holds $T$, and
all the other tokens stay in place.
Our objective is to determine the minimum number of swaps so that every
vertex $v_i$ holds the token $T_i$.
It is known (and easy to observe) that such a sequence of swaps 
always exists \cite{Yamanaka201581}.

We also study a generalized problem called the \emph{colored token swapping problem}.  Here every token and every vertex is given a color and 
we want to find a minimum sequence of swaps such that every token arrives on 
a vertex with the same color. In case that every token and every vertex is given a unique color, this is equivalent to the uncolored version.
%
This problem was introduced
in its full generality by
Yamanaka, E.~Demaine, Ito, Kawahara, Kiyomi, Okamoto, Saitoh, Suzuki,
Uchizawa, and Uno
~\cite{Yamanaka201581},
but special cases had been studied before.
When the graph is a path, the problem is equivalent to counting the number
of adjacent swaps in bubble sort, and it is folklore that this is exactly 
the number of inversions of the permutation $\pi$ (see Knuth \cite[Section 5.2.2]{KnuthTAOCP3}).
When the graph is complete, it was already known by Cayley \cite{cayley} that the minimum
number is equal to $n$ minus the number of cycles in $\pi$ (see also \cite{Jerrum1985265}).
Note that the number of inversions and the number of cycles can be computed in
$O(n \log n)$ time.
Thus, the minimum number of swaps can be computed in $O(n\log n)$ time
for paths and complete graphs.
Jerrum~\cite{Jerrum1985265} gave an $O(n^2)$-time algorithm to solve the 
problem for cycles.
When the graph is a star, a result of Pak~\cite{Pak1999329} implies 
an $O(n\log n)$-time algorithm.
Exact polynomial-time algorithms are also known for complete bipartite
graphs \cite{Yamanaka201581} and complete split graphs~\cite{Yasui-SIGAL}.
Polynomial-time approximation algorithms are known for 
trees with approximation factor two \cite{Yamanaka201581} and for squares of
paths with factor two~\cite{HeathV03}.
Since Yamanaka et al.\ \cite{Yamanaka201581}, 
it has remained open whether the problem is polynomial-time solvable
or NP-complete, even for general graphs, and whether there exists
a constant-factor polynomial-time approximation algorithm for general graphs.

Bonnet, Miltzow, and Rz\k{a}\.zewski~\cite{dblpTOKEN2}
have recently strengthened our
 results: 
 Deciding whether
 a token swapping problem on a graph with $n$ vertices has a solution 
 with $k$~swaps
 is W[1]-hard with respect to the parameter~$k$, and it cannot be solved in
 $n^{o(k/\log k)}$ time unless the ETH fails.
In addition, they show NP-hardness on graphs with treewidth at most~$2$.
Finally, they complement their lower bounds by showing that token swapping (and some variants of it) are FPT on nowhere dense graphs. 
This includes planar graphs and graphs of bounded
treewidth (still with the number $k$ of swaps as the parameter).

\paragraph*{Our results.}
In this paper, we present upper and lower bounds on the algorithmic 
complexity of the token swapping problem.
For exact computation, we derive a straightforward super-exponential time 
algorithm in Section~\ref{sec:SimpleExact}.
We also analyze the algorithm under the parameter $k$ and $\Delta$, where $k$ is 
the number of required swaps and $\Delta$ is the maximum degree of the underlying graph.

\begin{restatable}[Simple Exact Algorithm]{theorem}{SimpleAlgo}\label{thm:SimpleAlgo}
Consider the token swapping problem on a graph $G
$ with $n$ vertices and $m$ edges. Denote by $k$ the optimal
 number of required swaps and by $\Delta$ the maximum degree of $G$. 
An optimal sequence can be found within the following running times:
 \begin{enumerate}
  \item $O(n!\, mn\log n)=2^{O(n\log n)}$
  \item $O(m^{k}n\log n)=O(n^{2k+1}\log n)$
  \item $(2k\Delta)^{k}\cdot O(n\log n)$
 \end{enumerate}
The space requirement of these algorithms is
equal to the runtime limit divided by~$n\log n$.
\end{restatable}

With more problem specific insights, we could derive polynomial 
time constant factor approximation algorithms, see Section~\ref{sec:Approx1}. 
This resolves the first open problem
by Yamanaka et al.\ \cite{Yamanaka201581}.
We describe the algorithm for the uncolored version in Section~\ref{sec:Approx1}. 
In Section~\ref{relaxed}, we show a general 
technique to get also approximation algorithms for the more general colored version.
The approximation algorithm we suggest makes deep use of the 
structure of the problem and is nice to present. 

\begin{restatable}[Approximation Algorithm]{theorem}{Approximation}\label{thm:Approximation}
	There is a $4$-approximation of the colored token swapping problem on general graphs and 
  a $2$-approximation on trees.
\end{restatable}

Our main aim was to complement these algorithmic findings with corresponding conditional lower bounds. 
For this purpose, we design a gadget called
an \emph{even permutation network}\footnote{Not to be confused with a 
sorting network.}, in Section~\ref{sec:permutation}. This can be regarded as a special class of instances of the token 
swapping problem. It allows us to achieve \emph{any}
even permutation of the tokens between dedicated input and output vertices.
Using even permutation networks allows us to reduce further from
the colored token swapping problem to the more specific uncolored version.
Even permutation networks together with the NP-hardness proof for the \emph{colored} 
token swapping problem implies NP-hardness of the uncolored token swapping problem.
Yamanaka et al.\ \cite{Yamakana-2015} already proved that the
colored token swapping problem is NP-complete.
This gives a comparably simple way 
to answer the open question by Yamanaka et al.\ \cite{Yamakana-2015}.
However, this reduction has some polynomial blow up and thus does not give tight lower bounds.
We believe that even permutation networks are of general interest and 
we hope that similar 
gadgets will find applications in other situations.

Our reductions are quite technical and, so, we try to modularize 
them as much as possible. In the process, we show hardness
for a number of intermediate problems.
In Section~\ref{sec:disjoint}, we show a \emph{linear} reduction from $3$SAT to a problem of finding disjoint paths in a structured graph. (A precise definition can be found in that section.)
Section~\ref{sec:RedCloredTokens} shows how to reduce further to the colored token swapping problem.
Section~\ref{sec:ReductionTSP} is committed to the reduction from the colored to the uncolored token swapping problem. For this purpose, we attach 
an even permutation network to each color class. 
The blow up remains linear  as each color class has \emph{constant} size.
Section~\ref{sec:Finish} finally puts the 
three previous reductions together in order to attain the main results. 
For our lower bounds we need to use the exponential time hypothesis (ETH), 
which essentially states that there is no $2^{o(N)}$ algorithm for $3$SAT, 
where $N$ denotes the number of variables.

\begin{restatable}[Lower Bounds]{theorem}{LowerBounds}\label{thm:LowerBound}
	The token swapping problem has the following properties:
	\begin{enumerate}
	  \item\label{itm:NP} It is NP-complete.
	  \item\label{itm:ETH} It cannot be solved in time $2^{o(n)}$ unless the Exponential Time Hypothesis fails, where $n$ is the number of vertices.
	  \item\label{itm:Approx} It is APX-hard. 
	\end{enumerate}
	These properties also hold, when we restrict ourselves to instances of bounded degree. 
\end{restatable}

Therefore, our algorithmic results are almost \emph{tight}.
Even though our exact algorithm is completely straightforward 
our reductions show that we cannot hope for much better results.
The algorithm has complexity $2^{O(n\log n)}$ and we can rule out $2^{o(n)}$ algorithms.
Thus there remains only a $\log n$-factor in the exponent remaining as a gap.

In addition, we provide a $4$-approximation algorithm and show that there is a small constant $c$ such that there exists no $c$-approximation algorithm.
This determines the approximability status up to the constant $c$.
We want to point out that 
a crude estimate for the constant $c$ in Theorem~\ref{thm:LowerBound} is $c \approx 1+\frac{1}{1000}$. We do not believe that it is worth to compute $c$ exactly.
Instead, we hope that future research might find reductions with better constants.
Note that all our lower bounds for the uncolored token swapping problem carry over immediately to the colored version.

\paragraph*{Related concepts.}
The famous $15$-puzzle is similar to the token swapping
problem, and indeed a graph-theoretic generalization of the $15$-puzzle 
was studied by Wilson \cite{WILSON197486}.
The $15$-puzzle can be modeled as a token-swapping problem by
regarding the empty square of the $4\times 4$ grid as
a distinguished token, but there remains
an important difference from our problem: in the $15$-puzzle,
two adjacent
tokens can be swapped only if one of them is the distinguished token.
For the $15$-puzzle and its generalization, the reachability question is
not trivial, but by the result of Wilson \cite{WILSON197486},
it can be decided 
in polynomial
time.
However, it is NP-hard to find the minimum number of swaps even for
grids~\cite{Ratner1990111}.

The token swapping problem can be seen as an instance of the minimum
generator sequence problem of permutation groups.  There, a
permutation group is given by a set of generators $\pi_1,\dots,\pi_k$,
and we want to find a shortest sequence of generators whose
composition is equal to a given target permutation $\tau$.  This is
the problem that Jerrum \cite{Jerrum1985265} studied.  He gave an
$O(n^2)$-time algorithm for the token swapping problem on cycles.  He
proved that the minimum generator sequence problem is PSPACE-complete
\cite{Jerrum1985265}.  The token swapping problem is the special
case 
when the set of generators consists only of transpositions, namely
those transpositions that correspond to the edges of~$G$.

In the literature, we also find the problem of \emph{token sliding}, but
this is different from the token swapping problem.
In the token sliding problem on a graph $G$, we are given two independent sets
$I_1$ and $I_2$ of $G$ of the same size.
We place one token on each vertex of $I_1$, and we perfom a sequence of
the following sliding operations: 
We may move a token on a vertex $v$ to another vertex $u$ if 
$v$ and $u$ are adjacent, 
$u$ has no token,
and after the movement the set of vertices with tokens forms
an independent set of $G$.
The goal is to determine if a sequence of sliding operations can move
the tokens on $I_1$ to $I_2$.
The problem was introduced by Hearn and Demaine \cite{Hearn200572}, and 
they proved that the problem is PSPACE-complete even for planar graphs
of maximum degree three.
Subsequent research showed that the token sliding problem is 
PSPACE-complete for perfect graphs \cite{Kaminski20129}
and graphs of bounded treewidth \cite{MouawadIPEC14}.
Polynomial-time algorithms are known for cographs \cite{Kaminski20129}, 
claw-free graphs \cite{BonsmaSWAT14}, 
trees \cite{Demaine2015132}
and bipartite permutation graphs \cite{FoxEpsteinISAAC15}.

Several other models of swapping have been studied in the
literature 
 \cite{GrafESA15, tokengraphs, doi:10.1137/060652063}.

\section{Simple Exact Algorithms}\label{sec:SimpleExact}
We start presenting our results with the simplest one. There is an exact, exponential time algorithm which, after the hardness results that we obtain
in Section~\ref{sec:Finish}, will prove to be almost tight to the lower bounds.

\SimpleAlgo*

\begin{proof}
  The algorithm is breadth-first search in the configuration graph. 
The nodes $\mathcal{V}$ of the configuration graph $\mathcal{G} =
  (\mathcal{V},\mathcal{E})$ consist of all  $n! = 2^{O(n\log n)}$
  possible configurations of tokens on the vertices 
of~$G$. 
  Two configurations $A$ and $B$ are adjacent if 
  there is a single swap that transforms $A$ to $B$.
  Each configuration has $m\le\binom{n}{2} = O(n^2)$ adjacent configurations. Thus the total
  number of edges is $n!m/2$. 
  A shortest path in $\mathcal{G}$ from the start to the target
  configuration gives a shortest sequence of swaps. Breadth-first
  search finds this shortest path.

  We assume that the graph fits into main memory and its nodes can be
  addressed in constant time.  We can maintain the set of visited
  nodes (permutations) in a compressed binary trie, where searches,
  insertions, and deletions can be done in $O(\log n!) = O(n\log n)$
  time.  The running time of breadth-first search is thus
  $O(n\log n)\cdot(|\mathcal{V}| + |\mathcal{E}|) = O(n!mn\log n)=2^{O(n\log
    n)}$. This establishes the first claim.
  
  In case that the number of required swaps is $k$, breadth-first
  search explores the configuration graph for only $k$ steps. 
 Each configuration has at most $m = O(n^2)$ neighbors. Thus the total number of explored configurations is bounded by $m^{k} \leq n^{2k}$. This implies the second claim.
  
  For the the third claim, we need the observation that there is an optimal sequence of swaps that swaps only misplaced tokens. Assume that $S = (s_1,\ldots,s_k)$ is an optimal sequence of swaps and $s_i$ is a swap, where two tokens are swapped that were at the correct position before the swap. We construct an optimal sequence $S'=(s_1',\ldots,s_k')$, with one less swap of correctly placed tokens. The new sequence $S'$ skips the swap $s_i$ and adds it to the very end. It is easy to check that $S'$ is indeed a valid sequence of swaps and has one less swap of already correctly placed tokens.  Note that if $k$ swaps are sufficient to bring the tokens into the target configuration, at most $2k$ tokens are misplaced.
  
  If the maximum degree of the graph is bounded by $\Delta$, then, by the
  aforementioned observation, the breadth-first search needs to branch
  in at most $2k
\Delta$ new configurations from each
  configuration. This implies the third claim.
\end{proof}

\section{A 4-Approximation Algorithm} \label{sec:Approx1}

We describe an approximation algorithm
for the token swapping problem, which has approximation factor 4
on general graphs and 2 on trees. First, we state the following 
simple lemma.
\begin{lemma}
  \label{lem:lower-bound}
Let $d(T_i)$ be the distance of token $T_i$ to the target vertex $i$.
Let $L$ be the sum of distances of all tokens to their target vertices:
 \begin{equation*}
L := \sum_{i=1}^n{d(T_i)}.  
 \end{equation*}
Then any solution needs at least $L/2$ swaps.
\end{lemma}
\begin{proof}
 Every swap reduces $L$ by at most $2$.
\end{proof}

We are now ready to describe our algorithm. It has two atomic operations. 
The first one is
called an \emph{unhappy swap}: This is an edge swap where 
one of the tokens swapped is already on its target and the other token reduces
its distance to its target vertex (by one).

The second operation is called a \emph{happy swap chain}. Consider a path of 
$\ell+1$ distinct vertices $v_1,\ldots, v_{\ell+1}$. We swap the tokens over edge 
$(v_1,v_2)$, then $(v_2, v_3)$, etc., performing
$\ell$ swaps in total.
 The result is that the token that was 
on vertex $v_1$ is now on vertex $v_{\ell+1}$ and all other tokens have moved
from $v_j$ to $v_{j-1}$.
If every 
swapped token reduces its distance by at least $1$, we call this a happy swap chain of length $\ell$.
Note that a happy swap chain could consist of a single swap on a single edge.
A single swap that is part of a happy swap chain is called a \emph{happy swap}.
When our algorithm applies a happy swap chain, there will be an edge between
$v_{\ell+1}$ and $v_1$, closing a cycle. In this case, the happy swap chain
performs a cyclic shift.
Figure~\ref{fig:hCS} illustrates this definition with an example.
Note that a happy swap chain might consist of a single swap. We leave it to the reader to find such an example.

\begin{figure}[htb] 
 	  \centering\includegraphics[width = 0.7\textwidth]{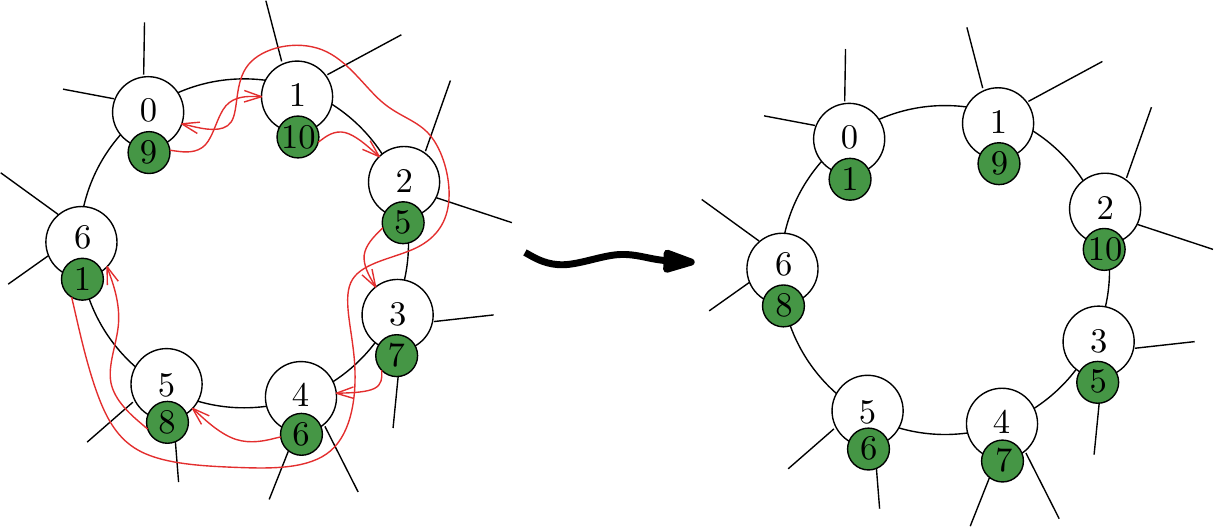}
  \caption{Before and after a happy swap chain. The swap sequence is, in this order, $6-5$, $5-4$, $4-3$, $3-2$, $2-1$, $1-0$. 
  Token $1$ swaps with every other token, moving counter-clockwise; every other token moves one step clockwise.}\label{fig:hCS}
\end{figure}


\begin{lemma}\label{lem:happyChain}
	Let $G=(V,E)$ be an undirected graph with a token placement
        where not every token is at its target vertex. Then there is a happy swap chain or an unhappy swap.
\end{lemma}
\begin{proof}
	Given a token placement on $G=(V,E)$, we define the directed graph $F$ on $V$ as follows. 
	For each undirected edge $e = \{v,w\}$ of $G$, we include the
        directed edge $(v,w)$ in~$F$ if the token 
	on $v$ reduces its distance to its target vertex by swapping along $e$. 
	Note that for a pair of vertices, both directed edges might be part of~$F$.
	We can perform a happy swap chain whenever we find a directed cycle in $F$.
	The outdegree of a vertex $v$ in $F$ is $0$ if and only if the token on $v$ has target vertex $v$. 
	Assume that not every token is in its target position. Choose any vertex $v$ that does not hold the 
	right token and construct a directed path from $v$ by
        following the directed edges of~$F$. This procedure will either revisiting a 
	vertex, and we get a directed cycle, or we encounter a vertex with outdegree $0$, and we get an unhappy swap.
\end{proof}

The lemma gives rise to our algorithm:
Search for a happy swap chain or unhappy swap; when one is 
found it is performed, until none remains.
If there is no such swap, the final placement of 
every token is reached.
This algorithm is polynomial time (follows from the proof of Lemma~\ref{lem:happyChain}). 
Moreover, it 
correctly swaps the tokens to their 
target position with at most $2L$ swaps.

\begin{lemma}\label{lem:noTwoUnhappy}
Let $T_i$ be a token on vertex $i$. If $T_i$ participates in an unhappy swap, then the next swap involving 
$T_i$ will be a happy swap. 
\end{lemma}
\begin{proof}
 Refer to Figure~\ref{fig:uhS}.	Let the vertices $i,j$ be the ones participating in the unhappy swap and let 
 $e$ be the edge that connects them.
 On vertex $j$ is token $T_i$ that got unhappily removed from its target vertex
 and on vertex $j$ is token $S'$ whose target is neither $i$ nor $j$.
 Based on Lemma~\ref{lem:happyChain}, our algorithm performs either unhappy swaps or happy swap chains.
 Note that, currently, edge $e$ cannot participate in an unhappy swap. This is because none 
 of its endpoints holds the right token. Moreover, token $T_i$ cannot participate in an unhappy
 swap that does not involve edge $e$, as that would not decrease its distance. 
 Therefore, there has to be a happy swap chain that involves token $T_i$.
\end{proof}

\begin{figure}[htb]
  \centering\includegraphics[width = 0.7\textwidth]{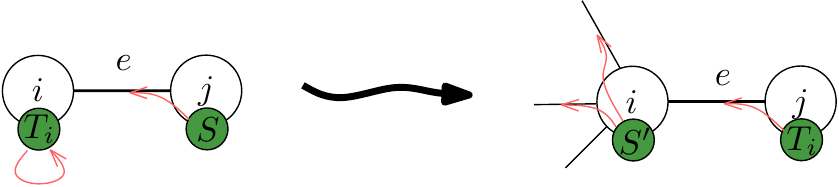}
  \caption{After a Token $T_i$ makes an unhappy swap along edge $e$, $T_i$ wants to go back to vertex $i$ and is not willing to go to any other vertex. Also whatever token $S'$ will be on vertex $i$, it has no desire to stay there. This implies that the next swap involving $T_i$ will be part of a happy swap chain.}\label{fig:uhS}
\end{figure}

\begin{theorem} \label{thm:approxration}
	Any sequence of happy swap chains and unhappy swaps is at most $4$ times as long as an optimal sequence
	 of swaps on general graphs and $2$ times as long on trees.
\end{theorem}
\begin{proof}
Let $L$ be the sum of all distances of tokens to its target
vertex. We know that the optimal solution needs at least $L/2$
swaps as every swap reduces $L$ by at most $2$ (Lemma~\ref{lem:lower-bound}). We will show that our algorithm needs at most $2L$ swaps, and this implies the claim.

A happy swap chain of length $l$ reduces the total sum of distances by $l+1$. Thus 
\[\#(\mbox{happy swaps}) < L.\]
By Lemma~\ref{lem:noTwoUnhappy},
\[\#(\mbox{unhappy swaps}) \leq \#(\mbox{happy swaps}),\] and this implies
\[ \#(\mbox{swaps}) = \#(\mbox{unhappy swaps}) + \#(\mbox{happy swaps}) \leq 2\cdot \#(\mbox{happy swaps}) < 2L. \]
 On trees, this algorithm is a $2$-approximation algorithm, as the
 longest possible cycle in $F$ (as in Lemma~\ref{lem:happyChain}) has
 length $2$, and thus every happy swap reduces $L$ by
 \emph{two}. 
This implies \[\#(\mbox{happy swaps}) = L/2. \qedhere\]

\end{proof}

\section{The Colored Version of the Problem}
\label{relaxed}

In this section, we consider the version of the problem when some tokens are
indistinguishable.
More precisely, each token $T_i$ has a {color} $C_i$,
each vertex $j$ has a color $D_j$, and the goal is to
let each token arrive at a vertex of its own color.
We call this the \emph{colored token swapping problem}. 
Of course we have to assume that the number of tokens of each color is
the same as the number of vertices of that color.

We will see that all approximation bounds from the previous
sections carry over to this problem.
Our approach to the problem is easy:
\begin{enumerate}
\item We first decide which token goes to which target vertex.
\label{matching}
\item We then apply one of the algorithms from the previous sections
  for $n$ distinct tokens.
\label{go}
\end{enumerate}

For Step~\ref{matching}, we solve a bipartite minimum-cost matching
problem based on distances in the graph: 
For each color $k$ separately, we set up
a complete bipartite graph between the tokens
 $T_i$ of color $k$ and the vertices $j$ of color $k$.
The cost $d_{ij}$ of each arc is the graph distance from
the starting position of token $T_i$ to vertex $j$.
We solve the assignment problem with these costs and
obtain a perfect matching 
between tokens 
and vertices.
We put the optimal matchings for the different colors together and
get a bijection 
$\pi\colon \{T_1,\dots,T_n\}\to V$
between all tokens and all vertices, of total cost
$$L^* = \sum_{i=1}^n d_{i\pi(i)}.$$

Finally, for Step~\ref{go}, we renumber each vertex $\pi(T_i)$ to $i$
and apply an algorithm for distinct tokens, moving
token $T_i$ to the vertex whose original number is $\pi(T_i)$.

The reason why the approximation bounds from the previous sections
carry over to this setting is that they are based on the lower bound
in Lemma~\ref{lem:lower-bound} from the sum of the graph distances.

\begin{lemma}
Consider any (not necessarily optimal) swap sequence 
$S$ that solves the colored version of the problem,
and let $L$ be the sum of the distances between each token's initial
and final position. Then $L\ge L^*$.
\end{lemma}
\begin{proof}
Since $S$ solves the problem, it must move each token $T_i$ to some
vertex $\sigma(T_i)$ of the same color, defined by an assignment 
$\sigma\colon \{T_1,\dots,T_n\}\to V$ which might be different than $\pi$.  
The sum of the distances then is 
\begin{equation}
  \label{eq:L}
L = \sum_{i=1}^n d_{i\sigma(i)},  
\end{equation}
 must satisfy  $L\ge L^*$, since the assignment $\pi$ was constructed
as the assignment minimizing the expression~\eqref{eq:L} among
all assignments $\sigma$ that respect the colors.
\end{proof}

The assignment problem on a graph with $N+N$ nodes and arbitrary costs
can be solved in $O(N^3)$ time. This must be summed over all colors,
giving a total time bound of $O(n^3)$ for Step~\ref{matching}. If the
color classes are small, the running time is of course
better. Depending on the graph class, the running time may also be
reduced. For example, on a path, the optimal matching can be
determined in linear time, assuming that the colors are consecutive
integers.

\Approximation*


\section{Even Permutation Networks}\label{sec:permutation}
Before we dive in the details of the reductions for token swapping (Sections~\ref{sec:disjoint}-\ref{sec:ReductionTSP}),
we first introduce the even permutation network. This is the gadget we have been advertising in the introduction. 
We consider it stand-alone and this is why we present it here, independently of the  other parts of the reduction.

Consider a family of token swapping instance $I(\pi)$, which depends on a permutation $\pi$ from some specific set of input vertices $V_\textup{in}$ to some set of output vertices $V_\textup{out}$. The permutation $\pi$ specifies, for each token initially on $V_\textup{in}$ the target vertex in $V_\textup{out}$. For every other token, the target vertex is independent of $\pi$. If the optimal number of swaps to solve $I(\pi)$ is the same for every \emph{even} permutation $\pi$, we call this family an \emph{even permutation network}, see Figure~\ref{even-perm}. We will refer to the permutation also as \emph{assignment}. Recall that a permutation is even, if the number of inversions of the permutation is even. Realizing \emph{all} permutations at the same cost is impossible, since every swap changes the parity, and hence all permutations reachable by a given number of swaps have the same parity.

We will use the even permutation networks later, to reduce from the \emph{colored} token swapping problem to the token swapping problem, see Lemma~\ref{lem:reductionTSP} and Figure~\ref{fig:attachPermuationNetwork}. The idea is to identify one color class with the input vertices of a permutation network. 
We use the simple trick of doubling the whole input graph before attaching the permutation networks. We thereby ensure that realizing even permutations is sufficient: the product of two permutations of the same parity is always even.

\begin{figure}[htbp]
  \centering
  \includegraphics{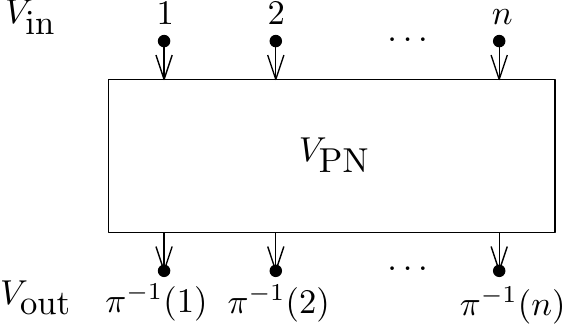}
  \caption{The interface of a permutation network PN}
  \label{even-perm}
\end{figure}

The constructed network has $n$ input vertices $V_\textup{in}$ and $n$ output
vertices $V_\textup{out}$, plus $O(n^3)$ inner vertices $V_{\mathrm{PN}}$.  For
each token on $ V_{\mathrm{PN}} \cup V_\textup{out}$, a fixed target vertex in
$V_{\textup{in}} \cup V_{\mathrm{PN}}$ is defined. The target vertices of~$V_\textup{in}$ lie in $V_\textup{out}$, but are still unspecified.
In this section, and only in this section, we refer to all tokens initially placed on $ V_{\mathrm{PN}} \cup V_\textup{out}$ as \emph{filler tokens}.


\begin{lemma}\label{lem:PermutationNetworks}
  For every $n$, there is a permutation network
$\mathrm{PN}
$ with $n$ input vertices $V_\textup{in}$, $n$ output
vertices $V_\textup{out}$,
and $O(n^3)$ additional vertices
$V_{\mathrm{PN}}$, which has the following properties, for some value $T$\textup:
\begin{itemize}
\item For every target assignment $\pi$ between the inputs $V_\textup{in}$ and the outputs $V_\textup{out}$ that is an even permutation, the shortest swapping sequence has
 length~$T$.
\item For any other target assignment, the shortest realizing sequence
  has length at least $T+1$.
\item The same statement holds for any extension
of the network PN,
 which is the union of PN with another graph $H$ that shares only the vertices $V_\textup{in}$ with PN, and in which the starting positions of the tokens assigned to $V_\textup{out}$ may be anywhere in $H$.
\end{itemize}
\end{lemma}
The last clause concerns not only all assignments between $V_\textup{in}$ and $V_\textup{out}$ that are odd permutations, but also all other conceivable situations where the tokens destined for $V_\textup{out}$ do not end up in $V_\textup{in}$, but somewhere else in the graph $H$. The lemma confirms that such non-optimal solutions for $H$ cannot be combined with solutions for PN to yield better swapping sequences than the ones for which the network was designed.

\begin{proof}
The even permutation network will be built up hierarchically from small gadgets.
Each gadget is built in a layered manner, subject to the following rules.
\begin{enumerate}
\item There is a strict layer structure: The vertices are partitioned
  into layers $U_1,\ldots,U_t$ of the same size.
\item Each gadget has its own input layer $V
=U_1$ at the top and its
  output layer $V'
=U_t$ at the bottom,
just as the overall network PN.
\item
\label{assign-unique}
 Edges may run
 between two vertices of
  the same layer
(\emph{horizontal  edges}), or between adjacent layers
(\emph{downward edges}).
\item Every vertex has at most 
 one neighbor in the successor layer 
and at most one neighbor in the predecessor layer.
\end{enumerate}
The goal is to bring the input tokens from the input layer $V_{\text{in}}$ to the output
layer~$V_{\text{out}}$. By Rule \ref{assign-unique}, the cheapest conceivable way to
achieve this 
is by using only the downward edges,
and then every such edge is used precisely once.

The gadget in Figure~\ref{fig:transpose} has an additional input
vertex $a$ with a fixed destination, indicated by the label~$a$.
This assignment forces us to use also horizontal edges, incurring some
extra cost as compared to letting the input tokens travel only
downward.

It can be checked easily that we only need to consider candidate 
swapping sequences in which every downward edge is used precisely once. 
This has the following consequences.
\begin{enumerate}
\item 
We can compare the cost of swapping sequences by counting the
horizontal edges used.
\item Each of the filler tokens,
which fill the vertices of the graph,
  moves one layer up to a fixed target location. 
  Thus we can predict the final
  position of all the auxiliary filler tokens that fill the network.
\end{enumerate}
Let us now look at the gadget of Figure~\ref{fig:transpose}.
It has three inputs 1,2,3 that connect the gadget with other gadgets, plus
an additional input $a$ with a specified target.
The initial vertex of $a$ and the target vertex of $a$ are not
connected to other parts of the graph.
The possible permutations of 1,2,3 that arise on the output are in the table
below the network, together with the number $\#h$ of used horizontal edges
(i.e., the relative cost).

Let us check this table.
To get the token $a$ into the correct target position, we must get it from the
fourth track to the third track, using (at least)
two horizontal edges ($AB$ or $DE$ or $AE$). With this minimum
cost of 2, we reach the identity permutation $1,2,3$. If we use edge
$C$, we can also reach the permutations $3,2,1$ (together with $A,B$)
and $2,1,3$ (together with $D,E$), at a cost of 3. All other possibilities
that bring the tokens from $V_{\text{in}}$ to $V_{\text{out}}$ and bring $a$ to the correct
target cost more, including ``crazy'' swapping sequences where the
tokens $1,2,3,a$ make intermediate upward moves.

In summary, the gadget can either realize the identity permutation, or
it can swap 1 with one of the other tokens at an extra cost of~1. 
We call this the \emph{swapping gadget}, and we will  symbolize it like 
in the right part of Figure~\ref{fig:transpose}, omitting the
auxiliary token~$a$.

\begin{figure}
  \centering
  \includegraphics{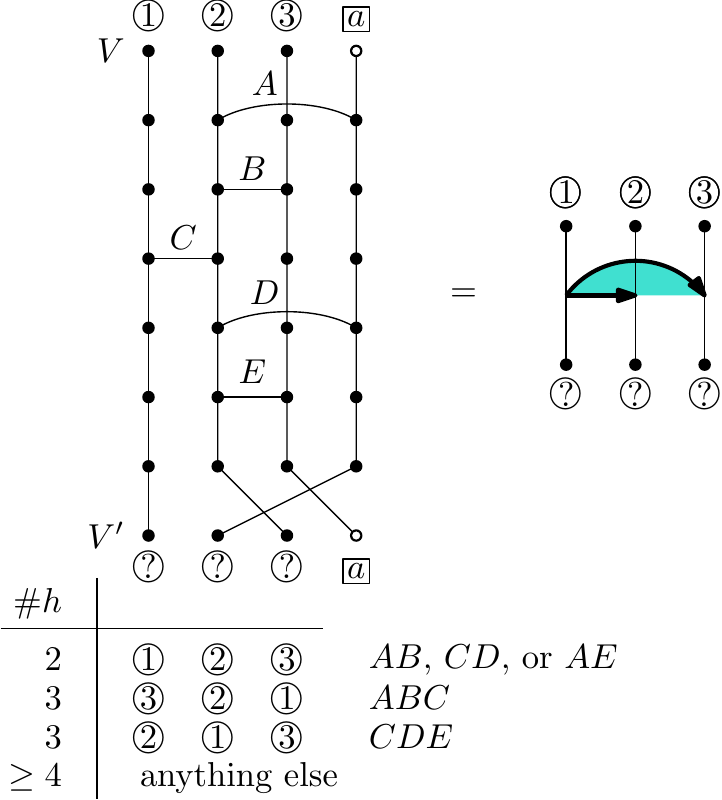}
  \caption{The gadget for swapping 1 with 2 or 3.  The number of
    horizontal edges used is denoted by $\#h$.}
  \label{fig:transpose}
\end{figure}

The next gadget is the \emph{shift gadget}, which is composed of two
swapping gadgets as shown in Figure~\ref{fig:shift}.
When putting together gadgets, the tracks must be filled up with
vertices on straight segments to maintain the strict layer structure, but we
don't show these trivial paths in the pictures.

We have a token $b$ with a fixed target, which is not connected to the
rest of the graph. To bring $b$ from the first track to the fourth
track and hence to the target, we cannot use the identity permutation
in 
 the 
 swap gadgets;
in both swapping gadgets, we must use one of the more expensive
swapping choices.
There are two choices that put $b$ on the correct target:
The choice ``swap $b\leftrightarrow2 $ and then swap $b\leftrightarrow 3$'' leads to the identity
permutation $1,2,3$.
The choice ``swap $b\leftrightarrow1$ and then swap
$b\leftrightarrow3$'' leads to the cyclic shift $3,1,2$.
We thus conclude that this gadget has two minimum-cost solutions:
The identity, or a cyclic shift to the right.
The shift gadget is symbolically drawn as a rectangular box as shown
in the right part of Figure~\ref{fig:shift},
indicating the optional cyclic shift.

\begin{figure}
  \centering
  \includegraphics{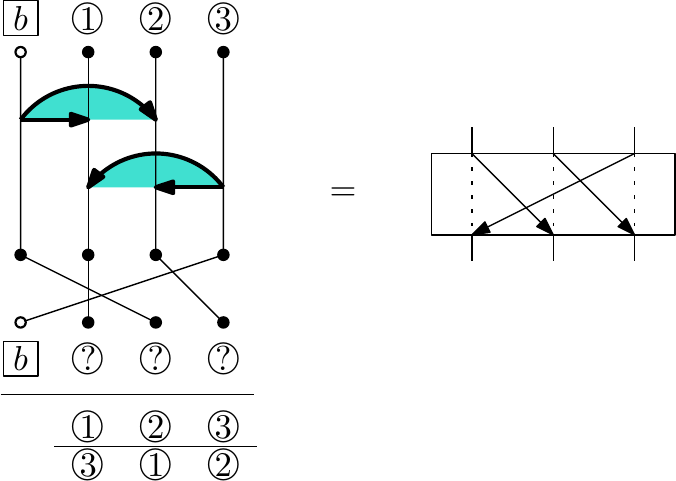}
  \caption{The gadget for an optional cyclic shift of 1,2,3 by one position}
  \label{fig:shift}
\end{figure}

From the shift gadget, we can now build the whole permutation network. 
The idea is similar to the realization of a permutation as a product of
adjacent transpositions, or to the bubble-sort sorting network.
Figure~\ref{fig:end} shows how a cascade of $n-1$ shift gadgets can
bring any input token to the last position.  We add a cascade of $n-2$
shift gadgets that can bring any of the remaining tokens to the
$(n-1)$-st position, and so on. In this way, we can bring any desired token
into the $n$-th, $(n-1)$-st, \dots, 3rd position. The first two
positions are then fixed, and thus we can realize half 
of all permutations,
namely the even ones.
\end{proof}

\begin{figure}
  \centering
  \includegraphics{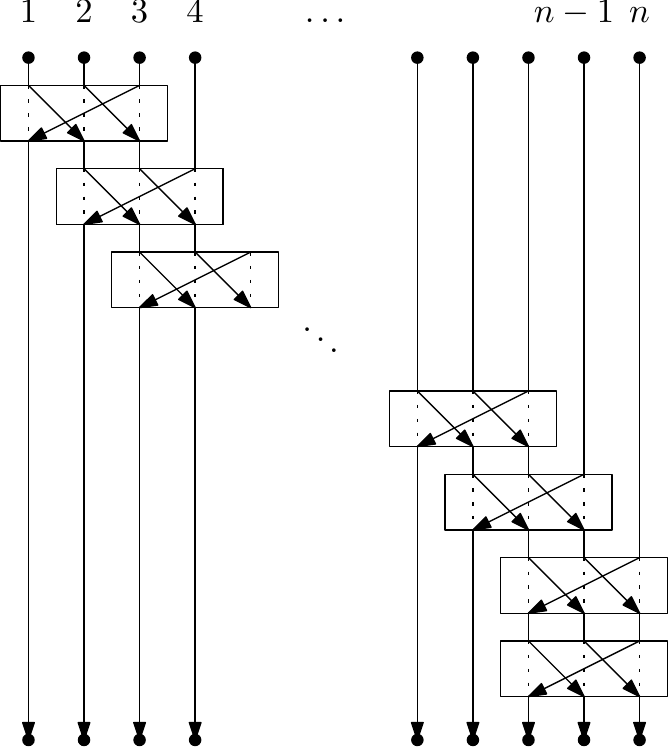}
  \caption{The gadget for bringing any desired token to the last position}
  \label{fig:end}
\end{figure}


\section{Reduction to a Disjoint Paths Problem}\label{sec:disjoint}

For the lower bounds of the Token Swapping problem, we study
some auxiliary problems. The first problem is a special
multi-commodity flow problem.
\begin{quote}
   \noindent \textbf{Disjoint Paths on a Directed Acyclic Graph (DP)}
 \\  \textbf{Input:} A directed acyclic graph $G=(V,E)$
and a bijection
 $\varphi \colon V^- \to V^+$
between the sources $V^-$ (vertices without incoming
arcs) and the sinks $V^+$ (vertices without outgoing arcs), with the
following properties:
 \begin{enumerate}
\item \label{layers}
The vertices can be partitioned
into {layers}
 $V_1,V_2,\ldots,V_t
$ such that,
 for every vertex 
in some layer $V_j$, all incoming arcs (if any)
come from the same layer $V_i$, with $i<j$.
Note that $i$ need not be the same for 
every vertex in $V_j$.
\item Every layer contains at most 10 vertices.
\item 
For every $v\in V^-$, there is a path from
 $v$ to $\varphi(v)$ in $G$.
Let $n(v)$ denote the number of vertices on the shortest path from
 $v$ to $\varphi(v)$.
\item \label{pack}
The total number of vertices is $
|V|= \sum_{v\in V^-} n(v)$.
 \end{enumerate}
 \textbf{Question:}
 Is there a set of vertex-disjoint directed paths 
 $P_1,\ldots,P_k$ with $k=|V^-|=|V^+|$, such that 
 $P_i$ starts at some vertex $v \in V^-$ and ends at $\varphi(v)$?
 \end{quote}
By Property \ref{pack}, the $k$ paths must completely cover the
vertices of the graph.
The graphs that we will construct in our reduction have in fact
the stronger property
that \emph{any} directed path from
$v\in V^-$ to $\varphi(v)$ 
contains the same number $n(v)$ of vertices.

   In our construction, we will label the source and the sink that
   should be connected by a path
   by the same symbol $X$, and ``the path $X$'' refers to this
   path.
In the drawings, the arcs will be directed from top to bottom.

 \begin{lemma}\label{lem:DPhard}
   There is a linear-size reduction from 3SAT to the
{Disjoint Paths Problem on a Directed Acyclic Graph (DP)}.
 \end{lemma}

 \begin{proof}
 Let $x_1, \ldots, x_n$ be the variables 
 and $C_1, \ldots, C_m$ the clauses of the 3SAT formula.
   Each variable $x_i$ is modeled by a variable path, which has a choice
   between two tracks.  
   The track is determined by the choice of the first vertex 
   on the path after $x_i$: either $x_i^T$ or $x_i^F$ . This choice 
   models the truth assignment.
 There is also a path for each clause.  In
   addition, there will be supplementary paths that fill the unused
   variable tracks.
Figure~\ref{fig:variable} shows an example of a variable $x_i$ that
appears in three clauses $C_u$, $C_v$, and $C_w$. 
Consequently, the
two tracks  $x_i^T$ and $x_i^F$, which run in parallel, 
pass through three \emph{clause} gadgets, which are shown schematically
as gray boxes in Figure~\ref{fig:variable} and which are drawn in greater detail in
Figure~\ref{fig:clause}.
The bold path in Figure~\ref{fig:variable} corresponds to assigning
the value \emph{false} to $x_i$: the path follows the track $x_i^F$.
For a variable that appears in $\ell$ clauses, there are $\ell+1$
supplementary paths. In
Figure~\ref{fig:variable}, they are labeled $s_1,\ldots,s_4$.
The path $s_j$ covers the unused track ($x_i^T$ in the example)
between the $(j-1)$-st and the $j$-th clause in which the variable
$x_i$ is involved. 

Initially, the path $x_i$ can choose between the tracks
$x_i^T$ and $x_i^F$; the other track will be covered by a path starting at $s_1$.
This choice is made possible by a \emph{crossing} gadget. Each variable 
has two crossing gadgets attached, one at the beginning of the variable path 
and one at the end. Those gadgets consist of 6 vertices: $x_i$, $s_1$,
$x_i^T$, $x_i^F$ and two auxiliary vertices that allow the variable to change
tracks. In Figure~\ref{fig:variable}, the crossing gadgets appear at the top
and the bottom. Note, that a variable can only change tracks in the 
crossing gadgets;
the last supplementary path 
$s_{\ell+1}$ allows the path $x_i$ to reach its target sink. 	

   \begin{figure}[htb]
     \centering
     \includegraphics{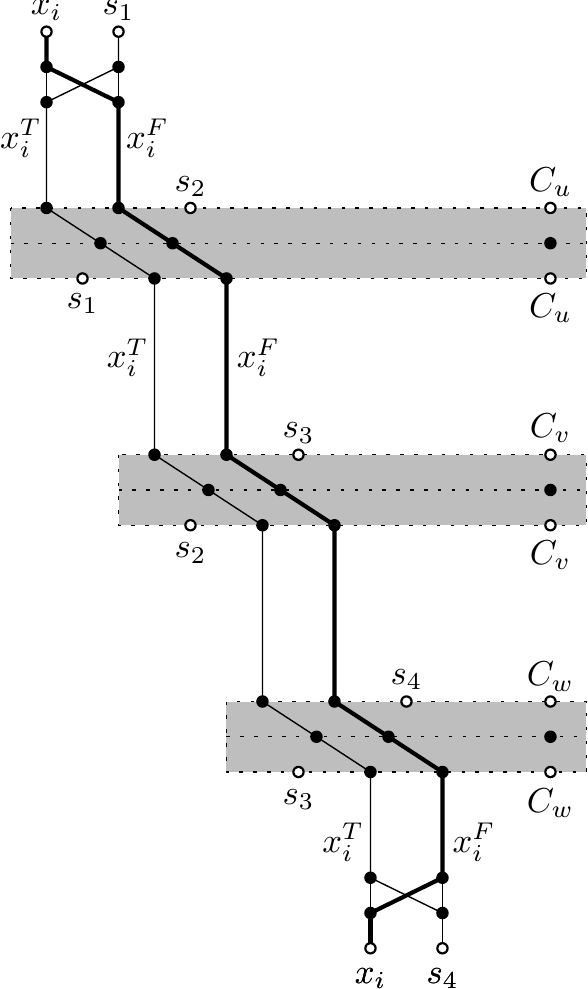}
     \caption{Schematic representation of a variable $x_i$. Sources and
       sinks are marked by white vertices, and their labels indicated the
       one-to-one correspondence $\varphi$ between sources and
       sinks. Arcs are directed from top to bottom.}
     \label{fig:variable}
   \end{figure}

   Figure~\ref{fig:clause} shows a clause gadget $C_z$ in greater
   detail. It consists of three successive layers and connects the
   three variables that occur in the clause. The clause itself is
   represented by a clause path that spans only these three layers. 
   A supplementary path starts at the first layer and one 
   ends at the third layer of each clause gadget. 
   Each layer of the clause gadget has at most 10 vertices: three for each variable, 
   two that are on the tracks $x_i^T,x_i^F$, one for the supplementary path of the variable.
   There is at most three variables per clause. Finally, there is a vertex that
   corresponds to the clause itself.
   
   Let $C_z$ be the current clause which is the $j$th clause in which $x_i$ appears,
   and it does so as a positive literal (as in Figure~\ref{fig:clause}a). 
   Each track, say $x_i^T$, connects
   to the corresponding vertex in the middle and bottom layer of the clause gadget
   and to $s_j$ (the end of the supplementary path). The track of 
   the literal that does not appear in this clause ($x_i^F$ in this case)
   is also connected to the vertex 
   of the clause on the middle layer ($z$ in Figure~\ref{fig:clause}). 
   Moreover, the middle layer vertex of this 
   track is connected to the top and bottom layer vertices that correspond to 
   the clause.
   The supplementary path $s_{j+1}$
   starts in this clause gadget and goes through the middle layer 
   and then connects to the vertices of both tracks on the bottom layer.
   Figure~\ref{fig:clause} depicts all these connections.

   \begin{figure}[htb]
     \centering
     \includegraphics[width = \textwidth]{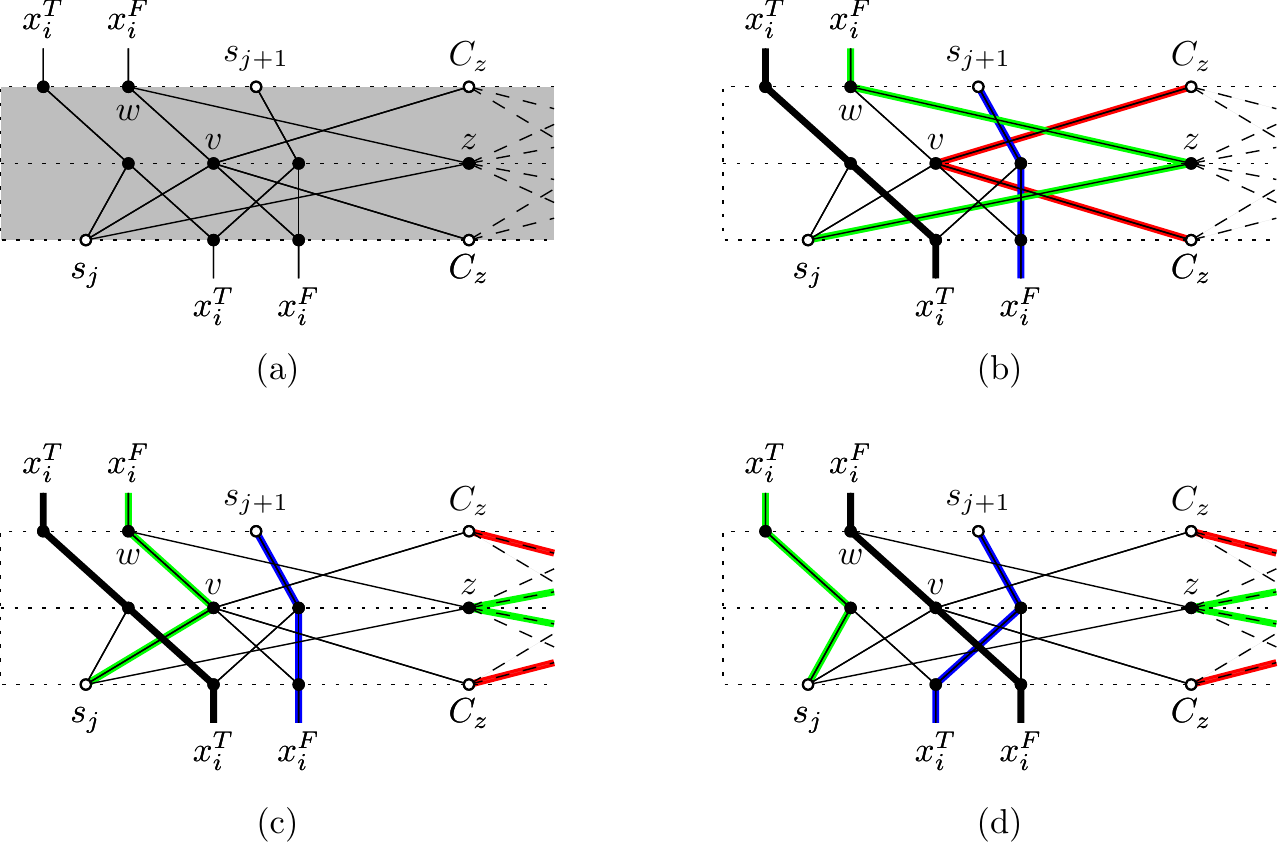}
     \caption{(a) Gadget for a clause $C_z$ containing a variable
       $x_i$ as a positive literal.
The clause involves two other variables, whose connections are
indicated by dashed lines. (b--d) The possibilities of paths passing
through the gadget:
 (b)~$x_i=\textit{true}$, the clause is fulfilled,
and the clause path makes its detour via the vertex $v$.
 (c)~$x_i=\textit{true}$, the clause is fulfilled,
and the clause path makes its detour via another vertex.
 (d)~$x_i=\textit{false}$, this variable does not contribute to
 fulfilling the clause,
and the clause path \emph{has to} make its detour via another vertex.
For a negative literal, the detour vertex $v$
and the upper neighbor $w$ of $z$
would be placed on the other track, $x_i^T$.
}
     \label{fig:clause}
   \end{figure}

We can make the following observations about the interaction between
the variable $x_i$ and the clause~$C_z$. We, further, illustrate those in 
Figure~\ref{fig:clause}.
\begin{enumerate}
\item 
A variable path (shown in black) that enters on the track $x_i^T$ or $x_i^F$ must leave
the gadget on the same track. Equivalently, a path cannot change track except in the 
crossing gadgets. 
\item
\label{done}
 The clause path (in red) can make a detour through vertex $v$ (w.r.t. Figure~\ref{fig:clause}) only if the
  variable $x_i$ makes the clause true, according to the track chosen by the variable path.
 In this case, the supplementary
  path $s_j$ covers the intermediate vertex $z$ of the clause.
\item 
If variable $x_i$ makes the clause true, the clause path may
also choose a different detour, in case more variables make the
clause true.
\item The supplementary path $s_j$ (shown in green) can reach its sink vertex.
\item The supplementary path $s_{j+1}$ (in blue) can reach the track
the track $x_i^T$ or $x_i^F$ which is not used by the variable path.
\end{enumerate}

Since each clause path \emph{must} make a detour into one of the variables,
it follows from Property~\ref{done}
that a set of disjoint source-sink paths exists if and only
if all clauses are satisfiable.

The special properties of the graph that are required for
the {Disjoint Paths Problem on a Directed Acyclic Graph (DP)} can be
checked easily. 
Whenever a vertex has 
one or more incoming arcs,
they come from the
previous layer of the same clause gadget.
We can assign three distinct layers to each clause gadget and to each
crossing gadget. 
Thereby, we ensure that every layer contains at most 10 vertices.
It follows from the construction that the graph has just enough
vertices
that the shortest source-sink paths can be disjointly packed, but we
can also check this explicitly.
If there are $n$ variables and $m$ clauses,
there are $30m+12n$ vertices in $G$:
30 per clause gadget, plus 12 for the two crossing gadgets of  each variable. 
For a variable $x_i$ that appears
in $\ell_i$ clauses, we have $n(x_i)=3(\ell_i+2)$. Each clause path $C_z$
contains $n(C_z)=3$ vertices, and each of the
 $\sum (\ell_i+1)=3m+n$ supplementary paths $s_j$ has
length
$n(s_j)=6$.
This gives  in total
$\sum n(v)
=3\sum \ell_i+6n+ 3m + 6(3m+n)
=9m+6n+ 3m + 18m+6n = 30m+12n=|V|$, as required.
 \end{proof}



\section{Reduction to Colored Token Swapping}\label{sec:RedCloredTokens}
We have shown in Lemma~\ref{lem:DPhard} how to reduce $3$SAT to the
{Disjoint Paths Problem on a Directed Acyclic Graph (DP)} in linear space. Now, we show how to reduce from this problem to the colored token swapping problem.

\begin{lemma}\label{lem:coloredhard}
	There exists a linear reduction from DP to the colored token swapping problem.
\end{lemma}
\begin{proof}
  Let $G$ be a directed graph and $\varphi$ be a bijection, as in the
  definition of DP.  We place $k$ tokens $t_1,\ldots,t_k$ of distinct
  colors on the vertices in $V^-$, see Figure~\ref{fig:ColorReduction}
  for an illustration. Their target positions are in $V^+$ as
  determined by the assignment $\varphi$.  We define a color for each
  layer $V_1,\ldots, V_{t-1}$.  Each vertex in layer $V_i$ which is
  not a sink is colored by the corresponding color. 
%
  Recall that for each vertex $v$ all ingoing edges come 
  from the same layer, which we denote by $L_v$.
  On each vertex $v$
  which is not a source, we place a token with the color of layer
  $L_v$. We call these tokens \emph{filler tokens}.
%
%
We set the threshold $T$ to 
 $|V(G)| - k$. This equals the number of filler tokens. 
	\begin{figure}[htbp]%
		\centering
		\includegraphics{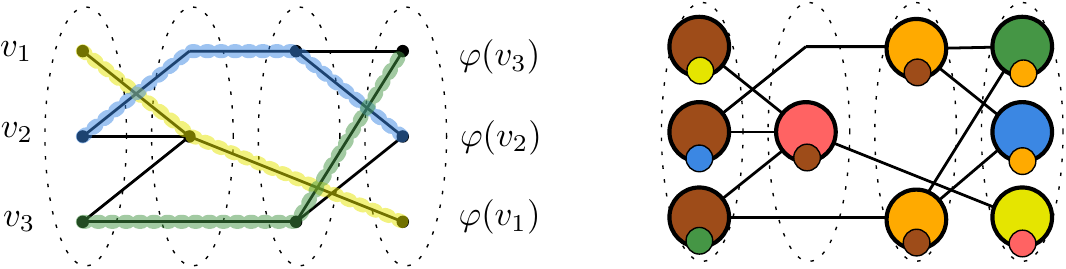}%
		\caption{Left: an instance of the disjoint paths
                  problem with $t=4$ layers
and $k=3$ source-sink pairs,
 together with a solution; arcs are directed from left to right.
Right: an equivalent instance of the colored token swapping problem.}%
		\label{fig:ColorReduction}%
\end{figure}
	
We have to show that there are $k$ paths with the properties above if and only if there is a sequence of at most $T$ swaps that brings every token to its target position.
	
[$\Rightarrow$] Given those paths, swap tokens $t_1,\ldots,t_k$ along their respective paths. In this way every token gets to its target vertex. (Recall that the paths $P_1,\ldots,P_k$ partition the set of vertices.)
	
[$\Leftarrow$] Assume we can swap every token to its target position
in $T$ swaps. Since each of the $T$ filler tokens must be swapped at
least once, we conclude that every filler token swaps exactly once.
In particular, every swap must swap a filler token
 and a non-filler token,
 which is a token that starts on a vertex in $V_1$,
and the filler token moves to a lower layer and the non-filler token
to a higher layer.

We denote by $P_1,\ldots,P_k$ the paths that the non-filler tokens
$t_1,\ldots,t_k$ follow. 
By the above argument,
 $P_1,\ldots,P_k$ must partition the vertex set.
The paths start and end at the correct position, because the tokens $t_1,\ldots,t_k$ do.
\end{proof}

We conclude that the Colored Token Swapping Problem is NP-hard. This
has already been proved by 
Yamanaka
, Horiyama, Kirkpatrick, Otachi, Saitoh, Uehara, and Uno
\cite{Yamakana-2015}, even when there are only three colors.

The reduction of Lemma~\ref{lem:coloredhard} produces instances of the colored token swapping problem with additional properties, which are directly derived from the properties of DP. For later usage, we summarize them in the following definition.

\begin{quote}
   \noindent \textbf{Structured Colored Token Swapping Problem}
 \\  \textbf{Input:} A number $k\in \mathbb{N}$ and a graph $G=(V,E)$, where each vertex has a not necessarily unique color. Further on each vertex sits a colored token. 
\begin{enumerate}
 \item There exists a partition of the vertices $V$ into layers $V_1,\ldots,V_t$ such that each layer has at most 10 vertices.
 \item Each edge can be oriented so that all outgoing edges of a vertex go to layers with larger index. 
 \item There exists a bijection $\varphi$ between the sources $V^-\subseteq V$ and the sinks $V^+\subseteq V$, such that for each vertex $v\in V^-$ the token on $v$  has target vertex $\varphi(v)\in V^+$.
 \item All non-sink vertices of layer $V_i\setminus V^+$ have the same color, for all $i=1,\ldots, t$.
 \item All non-source vertices $v\in V_i\setminus V^-$ have all ingoing edges to the same layer $V_j$ with $j<i$.
\end{enumerate}
\textbf{Question:} Does there exist a sequence of $k$ swaps such that each token is on a vertex with the same color?
\end{quote}

\begin{corollary}\label{cor:StrCTSW}
  There exists a linear reduction from DP to the structured colored token swapping problem.
\end{corollary}


\section{Reduction to the Token Swapping Problem}\label{sec:ReductionTSP}
In this section we describe the final reduction which results in an
instance of the token swapping problem. To achieve this we make use of
the even permutation network gadget from
Section~\ref{sec:permutation}.

\begin{lemma}\label{lem:reductionTSP}
	There exists a linear reduction from the structured colored token swapping problem  to the token swapping problem.
\end{lemma}
\begin{proof}
  Let $I$ be an instance of the structured colored token swapping instance. We denote the graph by $G$, the layers by $V_1,\ldots,V_t$, the sources by $V^-$, the sinks by $V^+$ and the threshold for the number of swaps by $k$. 

  We construct an instance $J$ of the token swapping problem. The graph $\overline{G}$ consists of two copies of $G$. For each set $V_j\setminus V^+$, we add \emph{one} permutation network to the union of both copies. In other words, the two copies of $V_j\setminus V^+$ serve as inputs of the permutation network. We denote the output vertices of the permutation network attached to the copies of $V_j\setminus V^+$ by $V_j'$.
  The filler tokens that were destined for $V_j\setminus V^+$ have $V_j'$ as their new final destination, see Figure~\ref{fig:attachPermuationNetwork}.
  It is not important how we assign each token to a target vertex, as long as this assignment corresponds to an \emph{even} permutation. 
  \begin{figure}[htbp]
  \centering
  \includegraphics{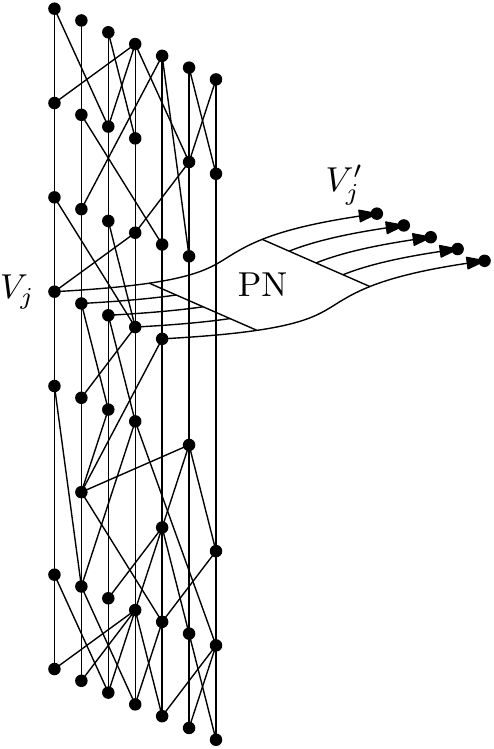}
  \caption{Attaching a permutation network to let the filler tokens
    destined for $V_{j}$ arrive at their final destinations $V_j'$.
The small ``vertical'' network in this figure represents the network constructed in Lemma~\ref{lem:coloredhard}.
Note that we did not duplicate $G$ in the figure.}
  \label{fig:attachPermuationNetwork}
\end{figure}
  Further each token gets a unique label. The threshold $\overline{k}$ for the number of swaps is defined by $2k$ plus the number of swaps needed for the permutation networks.
  
  We show at first that this reduction is linear. For this it is sufficient to observe that the size of each layer is constant. And thus also the permutation network attached to these layers have constant size each.
  
  Now we show correctness. Let $I$ be an instance of the structured colored token swapping problem and $J$ the constructed instance as described above.
  
  $[\Rightarrow]$ Let $S$ be a valid sequence of swaps that brings every token  on $G$ to a correct target position within $k$ swaps. We need to show that there is a sequence of $\overline{k}$ swaps that brings each token in $\overline{G}$ to its unique target position.  We perform $S$ on each copy of $G$. Thereafter every token is swapped through the permutation network to its target position. Note that the permutation between the input and output is an \emph{even} permutation. This is because we doubled the graph $G$. Therefore, the number of swaps in the permutation network is constant regardless of the permutation of the filler tokens in layer $V_i$, by Lemma~\ref{lem:PermutationNetworks}. This also implies that the total number of swaps is $\overline{k}$.
  
  $[\Leftarrow]$ Assume that there is a sequence $S'$ of $\overline{k}$ swaps that brings each token of $J$ to its target position. This implies that each filler token went through the permutation network to its correct target vertex. In order to do that each filler token must have gone to some input vertex of its corresponding permutation network. As the number of swaps inside each permutation network is independent of the permutation of the tokens on the input vertices, there remain exactly $2k$ swaps to put all tokens in each copy of $G$ at its right place. This implies that each token in $G$ can be swapped to its correct position in $I$ in $k$ swaps as claimed.
\end{proof}

\section{Hardness of the Token Swapping Problem}\label{sec:Finish}
In this section we put together all the reductions. They imply the following theorem.
\LowerBounds*
\begin{proof}
For the NP-completeness, we reduce from 3SAT. 
For the lower bound under ETH, we need to use the Sparsification Lemma, see~\cite{impagliazzo1999complexity}, and reduce from  3SAT instances where the number of clauses is linear in the number of variables. This prevents a potential quadratic blow up of the construction.
For the inapproximability result, we reduce from $5$-OCCURRENCE-MAX-$3$-SAT. (In this variant of $3$SAT each variable is allowed to have at most \emph{$5$ occurrences}.)
This gives us some additional structure that we use for the argument later on. In all three cases the reduction is exactly the same. 

Let $f$ be a 3SAT instance. We denote by $K(f)$ the instance of the token swapping problem after applying the reduction of Lemma~\ref{lem:DPhard}, Lemma~\ref{lem:coloredhard} and Lemma~\ref{lem:reductionTSP} in this order.
(It is easy to see that the graph of $K(f)$ has bounded degree.)

As all three reductions are correct, we can immediately conclude that the problem is NP-hard. NP-membership follows easily from the fact that a valid sequence of swaps is at most quadratic in the size of the input and can be checked in polynomial time.

The Exponential Time Hypothesis (ETH) asserts that 3SAT cannot be solved in $2^{o(n')}$, where $n'$ is the number of variables. The Sparsification Lemma implies that 3SAT cannot be solved in $2^{o(m')}$, where $m'$ is the number of clauses. As the reductions are linear the number of vertices in $K(f)$ is linear in the number of clauses of $f$. Thus a subexponential-time algorithm for the token swapping problem implies a subexponential-time algorithm for 3SAT and contradicts ETH.

To show APX-hardness, we do the same reductions as before, but we reduce from $5$-OCCURRENCE-MAX-$3$-SAT. Thus we can assume that each variable in $f$ appears in at most $5$ clauses. This variant of 3SAT is also APX-hard, see~\cite{arora1996hardness}. Assume a constant fraction of the clauses of $f$ are not satisfiable. We have to show that we need an additional constant fraction on the total number of swaps. For this, we assume that the reader is familiar with the proofs of Lemma~\ref{lem:DPhard},~\ref{lem:coloredhard} and~\ref{lem:reductionTSP}.
It follows from these proofs that there is a constant sized gadget in $K(f)$ for  each clause of $f$. Also there are certain tokens that represent variables and the paths they take correspond to a variable assignment.
We denote with $x,y,z$ the variables of some clause $C$, and we denote with $T_x,T_y,T_z$ the tokens corresponding to these variables and $G_C$ the gadget corresponding to $C$.
We need two crucial observations.
In case that the paths that the tokens $T_x, T_y$ and $T_z$ take do not correspond to an assignment that makes $C$ true, at least one more swap is needed that can be attributed to this clause gadget $G_C$.
In case that a token $T_x$ changes its track, which corresponds to another assignment of the variable, then at least one more swap needs to be performed that can be attributed to all its clauses with value $1/5$.
These two observations follow from the proofs of the above lemmas, as otherwise it would be possible to ``cheat'' at each clause gadget and the above lemmas would be incorrect.
The observations also imply the claim. Let $f$ be a 3SAT formula with a constant fraction of the clauses not satisfiable. Assume at first that the swaps are `honest' in the sense that the variable tokens $T_x$ does not change its track and corresponds consistently with the same assignment.
In this case, by the first observation, we need at least one extra swaps per clause. And thus a constant fraction of extra swaps, compared to the total number of swaps.
In a dishonest sequence of swaps, changing the 
track of some variable token $T_x$ fixes at most five clauses.
This implies at least one extra swap for every five unsatisfied clauses, which is a constant fraction of all the swaps as the total number of swaps is linear in the number of clauses of $f$. This finishes the APX-hardness proof.
\end{proof}

\paragraph*{Acknowledgments}
This project was initiated on the GREMO 2014 workshop in Switzerland. We want to thank the organizers for inviting us to the very enjoyable workshop.
Tillmann Miltzow is partially supported by the ERC grant PARAMTIGHT: ``Parameterized complexity and the search for tight complexity results'', no. 280152.
Yoshio Okamoto is partially supported by MEXT/JSPS KAKENHI Grant Numbers 24106005, 24700008, 24220003 and 15K00009, and JST, CREST, Foundation of Innovative Algorithms for Big Data.


%
\bibliographystyle{plainurl} 
\bibliography{Lib}
\goodbreak
\end{document}